\documentclass[cmp,final,envcountsect,envcountsame]{svjour}

\usepackage{amsmath,amssymb}
\usepackage{amsthm}
\usepackage{amsfonts}
\usepackage{color}
\usepackage{mathtools}
\usepackage{enumerate}
\usepackage{overpic}
\usepackage{bbm}

\usepackage[ colorlinks = true,
             linkcolor = blue,
             urlcolor  = blue,
             citecolor = red,
             anchorcolor = green,
]{hyperref}

\DeclareMathOperator{\tr}{Tr}
\DeclareMathOperator{\supp}{supp}
\DeclareMathOperator{\poly}{poly}
\newcommand{\cI}{\mathcal{I}}
\newcommand{\cR}{\mathcal{R}}
\newcommand{\cC}{\mathcal{C}}
\newcommand{\cS}{\mathcal{S}}
\newcommand{\cT}{\mathcal{T}}

\newcommand{\cM}{\mathcal{M}}

\newcommand{\cH}{\mathcal{H}}

\newcommand{\bN}{\mathbb{N}}
\newcommand{\eps}{\varepsilon}

\begin{document}

\title{On Composite Quantum Hypothesis Testing}

\author{Mario Berta \inst{1,2,3} \and Fernando G.~S.~L.~Brand\~ao \inst{2,3} \and Christoph Hirche \inst{4}}
\institute{Department of Computing, Imperial College London, London SW7 2AZ, UK \and
IQIM, California Institute of Technology, Pasadena, CA 91125, USA \and
AWS Center for Quantum Computing, Pasadena, CA 91125, USA \and 
QMATH, Department of Mathematical Sciences, University of Copenhagen, Denmark}

%F\'{\i}sica Te\`{o}rica: Informaci\'{o} i Fen\`{o}mens Qu\`{a}ntics, Departament de F\'{i}sica, Universitat \\ Aut\`{o}noma de Barcelona, Spain}

%\email{christoph.hirche@uab.cat}

\maketitle

\begin{abstract}
We extend quantum Stein's lemma in asymmetric quantum hypothesis testing to composite null and alternative hypotheses. As our main result, we show that the asymptotic error exponent for testing convex combinations of quantum states $\rho^{\otimes n}$ against convex combinations of quantum states $\sigma^{\otimes n}$ can be written as a regularized quantum relative entropy formula. We prove that in general such a regularization is needed but also discuss various settings where our formula as well as extensions thereof become single-letter. This includes an operational interpretation of the relative entropy of coherence in terms of hypothesis testing. For our proof, we start from the composite Stein's lemma for classical probability distributions and lift the result to the non-commutative setting by using elementary properties of quantum entropy. Finally, our findings also imply an improved recoverability lower bound on the conditional quantum mutual information in terms of the regularized quantum relative entropy\,---\,featuring an explicit and universal recovery map.
\end{abstract}

%%%%%%%%%%%%%%%%%%%%%%%%%%%%%%%%%%%%%%%%%%%%%%%%%%%%%%%%%%%%%%%%%%%%%%%%%%%%%%%%

\section{Overview of results}\label{sec:intro}
 
Hypothesis testing is arguably one of the most fundamental primitives in quantum information theory. As such it has found many applications, e.g., in quantum channel coding~\cite{HN03} and quantum illumination~\cite{Lloyd2008,Tan2008,Wilde17}, or for giving an operational interpretation to abstract quantities~\cite{BP10,HT14,CHMOSWW16}. A particular hypothesis testing setting is that of quantum state discrimination where quantum states are assigned to each of the hypotheses and we aim to determine which state is actually given. Several distinct scenarios are of interest, which differ in the priority given to different types of error or in how many copies of a system are given to aid the discrimination. Here, we investigate the setting of asymmetric hypothesis testing where the goal is to discriminate between two $n$-party quantum states (strategies or hypotheses) $\rho_n$ and $\sigma_n$ living on the $n$-fold tensor product of some finite-dimensional inner product space $\cH^{\otimes n}$. That is, we are optimizing over all two-outcome positive operator valued measures (POVMs) with $\{M_n,(1-M_n)\}$ and associate $M_n$ with accepting $\rho_n$ as well as $\left(1-M_n\right)$ with accepting $\sigma_n$. This naturally gives rise to the two possible errors
\begin{align}\label{eq:errorsI}
&\alpha_n(M_n):=\tr\big[\rho_n(1-M_n)\big]\; &\text{Type 1 error,} \\ 
&\beta_n(M_n):= \tr\big[\sigma_nM_n\big]\; &\text{Type 2 error.}\label{eq:errorsII}
\end{align}
For asymmetric hypothesis testing we minimize the Type 2 error as\footnote{Here and henceforth $\ll$ denotes the Loewner order.}
\begin{align}
\beta(n,\eps):=\inf_{0\ll M_n\ll1}\Big\{\beta_n(M_n)\big|\alpha_n(M_n)\leq\eps\Big\}
\end{align}
while we require the Type 1 error not to exceed a small constant $\eps\in(0,1)$. We are then interested in finding the optimal error exponent\footnote{Here and henceforth the logarithm is defined with respect to the basis 2.}
\begin{align}
\zeta(n,\eps):=-\frac{\log\beta(n,\eps)}{n},
\end{align}
and its asymptotic limits
\begin{align}\label{Eq:SteinsExponent}
\zeta(\infty,\eps):=\lim_{n\to\infty}-\frac{\log\beta(n,\eps)}{n},\quad \zeta(\infty,0):=\lim_{\eps\to0}\zeta(\infty,\eps)\,.
\end{align}
%(whenever the limits exist)

A well studied discrimination setting is that between fixed independent and identical (iid) states $\rho^{\otimes n}$ and $\sigma^{\otimes n}$, where the asymptotic error exponent is determined by the quantum Stein's lemma~\cite{HP91,ON00,Audenaert2008} in terms of the quantum relative entropy. Namely, we denote this special case of Equation \eqref{Eq:SteinsExponent} by $\zeta_{\rho,\sigma}(\infty,\eps)$ and the Stein's lemma then gives for any $\eps\in(0,1)$ the formula
\begin{align}\label{eq:brandao_old}
\zeta_{\rho,\sigma}(\infty,\eps)= D(\rho\|\sigma):=\begin{cases} \tr\big[\rho\left(\log\rho-\log\sigma\right)\big] \quad & \supp(\rho)\subseteq\supp(\sigma)\\ \infty & \text{otherwise.}\end{cases}
\end{align}
In many applications we aim to solve more general discrimination problems and a prominent example of such is that of composite hypotheses\,---\,in which we attempt to discriminate between different sets of states. Previously the case of composite iid null hypotheses $\rho^{\otimes n}$ with $\rho\in\cS$ and corresponding asymptotic error exponent $\zeta_{\cS,\sigma}(\infty,\eps)$ was studied in~\cite{hayashi2002,BDKSSS05}, leading to the formula
\begin{align}
\zeta_{\cS,\sigma}(\infty,\eps)=\inf_{\rho\in S} D(\rho\|\sigma)\quad\forall\eps\in(0,1)\,.
\end{align}
On the other hand, the problem of composite alternative hypotheses is more involved in the non-commutative case. When the set of alternative hypotheses $\cT_n$ for $n\in\mathbb{N}$ fulfils certain axioms motivated by the framework of resource theories, it was shown in~\cite{BP10} that the corresponding asymptotic error exponent $\zeta_{\rho,\cT}(\infty,\eps)$ is written in terms of the regularized relative entropy distance as
\begin{align}
\zeta_{\rho,\cT}(\infty,\eps)=\lim_{n\to\infty}\frac{1}{n}\inf_{\sigma_n\in\cT_n} D\left(\rho^{\otimes n}\|\sigma_n\right)\quad\forall\eps\in(0,1)\,.
\end{align}
This regularization is in general needed as we know from the case of the relative entropy of entanglement~\cite{Vollbrecht01}. Note that this might not be too surprising since the set of alternative hypotheses is not required to be iid in general.

For our main result, we consider the setting where null and alternative hypotheses are both composite and given by convex combinations of $n$-fold tensor powers of states from given convex, closed sets $\cS$ and $\cT$. %(see Section \ref{sec:mainresult} for the precise definition). 
More precisely, for $n\in\mathbb{N}$ we attempt the following discrimination problem.\footnote{Here and henceforth all inner product spaces $\cH$ are finite-dimensional and $S(\cH)$ denotes the set of unit trace positive semi-definite linear operators on $\cH$.}

\begin{description}
\item[{\bf Null hypothesis}:] the convex hull of iid states
\begin{align}
\text{$\cS_n:=\Big\{\int\rho^{\otimes n}\;\mathrm{d}\nu(\rho)\Big|\nu\in\cS\Big\}$ with $\cS\subseteq S(\cH)$ convex and closed}
\end{align}
\item[{\bf Alternative hypothesis}:] the convex hull of iid states
\begin{align}
\text{$\cT_n:=\Big\{\int\sigma^{\otimes n}\;\mathrm{d}\mu(\sigma)\Big|\mu\in\cT\Big\}$ with $\cT\subseteq S(\cH)$ convex and closed}
\end{align}
\end{description}

Slightly abusing the notation, $\nu\in\cS$ and $\mu\in\cT$ stand for probability measures on the Borel $\sigma$-algebra of $\cS$ and $\cT$, respectively. For $\eps\in(0,1)$ the goal is to determine the optimal error exponent for composite asymmetric hypothesis testing
\begin{align}
\zeta_{\cS,\cT}(n,\eps):=-\frac{1}{n}\log\inf_{0\ll M_n\ll1}\left\{\sup_{\mu\in\cT}\tr\big[M_n\sigma_n(\mu)\big]\middle|\sup_{\nu\in\cS}\tr\big[(1-M_n)\rho_n(\nu)\big]\leq\eps\right\}
\end{align}
%(as we will see the following limit exists)
%\begin{align}
%&\beta_{\cS,\cT}(n,\eps):=-\frac{1}{n}\log\inf_{0\leq M_n\leq1}\left\{\sup_{\mu\in\cT}\tr\left[M_n\sigma_n(\mu)\right]\middle|\sup_{\nu\in\cS}\tr\left[(1-M_n)\rho_n(\nu)\right]\leq\eps\right\}\\
%&\zeta_{\cS,\cT}(\eps):=\lim_{n\to\infty}\zeta_{\cS,\cT}^n(\eps)\quad\mathrm{and}\quad\zeta_{\cS,\cT}(0):=\lim_{\eps\to0}\zeta_{\cS,\cT}(\eps)\,,
%\end{align}
with the abbreviations
\begin{align}
\rho_n(\nu):=\int\rho^{\otimes n}\mathrm{d}\nu(\rho)\quad\text{and}\quad\sigma_n(\mu):=\int\sigma^{\otimes n}\mathrm{d}\mu(\sigma)\,.
\end{align}
It is trivial to see that $\zeta_{\cS,\cT}(n,\eps)$ equivalently gives the error exponent of testing between $\cS^{\otimes n} := \{ \rho^{\otimes n} | \rho\in\cS \}$ and $\cT^{\otimes n} := \{ \sigma^{\otimes n} | \sigma\in\cT \}$. This then explicitly takes the form of an iid problem. The following is our main result, which we prove in Section~\ref{sec:mainresult} under the support condition
\begin{align}
\mathrm{supp}(\rho) \subseteq \mathrm{supp}(\sigma)\quad\forall\rho\in\cS\quad\forall\sigma\in\cT\,.
\end{align}
%It is also worth pointing out that the same error exponent would hold if we would have defined the sets $\cS$ and $\cT$ simply as $\cS^{\otimes n}$ and $\cT^{\otimes n}$, respectively.

\begin{theorem}\label{thm:main}
For the discrimination problem as above, we have
\begin{align}\label{eq:liminf}
\lim_{\eps\to0}\liminf_{n\to\infty}\zeta_{\cS,\cT}(n,\eps)&=\lim_{\eps\to0}\limsup_{n\to\infty}\zeta_{\cS,\cT}(n,\eps) \\
&=\lim_{n\to\infty}\frac{1}{n}\inf_{\substack{\rho\in\cS\\\mu\in\cT}}D\Big(\rho^{\otimes n}\Big\|\int\sigma^{\otimes n}\;\mathrm{d}\mu(\sigma)\Big) \label{eq:Main}
%&=\sup_{n\in\bN}\frac{1}{n}\inf_{\substack{\rho\in\cS\\\mu\in\cT}}D\Big(\rho^{\otimes n}\Big\|\int\sigma^{\otimes n}\;\mathrm{d}\mu(\sigma)\Big) %\\
%\limsup_{n\to\infty}\frac{1}{n}\inf_{\substack{\rho\in\cS\\\mu\in\cT}}D\Big(\rho^{\otimes n}\Big\|\int\sigma^{\otimes n}\;\mathrm{d}\mu(\sigma)\Big)\,.\label{eq:limsup}
\end{align}
\end{theorem}

%We note that in general the existence of the limit $n\to\infty$ remains unclear\,---\,as in contrast to the work \cite{BP10} we do not necessarily have monotonicity in $n$ because of the missing $\sigma_m\in\cT_m,\;\sigma_n\in\cT_n\nRightarrow\sigma_m\otimes\sigma_n\in\cT_{mn}$. Nevertheless, in Section~\ref{sec:examples} we give several examples for which the limit inferior in Equation~\eqref{eq:liminf} and the limit superior in Equation~\eqref{eq:limsup} indeed coincide.
%We show that for $n\to\infty$ and $\eps\to0$ the error exponent $\zeta_{\cS,\cT}(\infty,0)$ is quantified by the formula
%\begin{align}\label{eq:main_first}
%\frac{1}{n}\inf_{\substack{\rho\in\cS\\\mu\in\cT}}D\Big(\rho^{\otimes n}\Big\|\int\sigma^{\otimes n}\;\mathrm{d}\mu(\sigma)\Big).
%\end{align}
%Note that we slightly misused the notation with $\mu\in\cT$ standing for probability measures on the $\sigma$-algebra generated by all singletons in the power set of $\cT$.

Our proof can be found in Section~\ref{sec:mainresult} and has a clear structure in the sense that we start from the composite Stein's lemma for classical probability distributions and then lift the result to the non-commutative setting by using elementary properties of entropic measures. We emphasise that even in the case of a fixed null hypothesis $\cS=\{\rho\}$ our setting is not a special case of the previous results~\cite{BP10}, as our sets of alternative hypotheses are not closed under tensor product
\begin{align}
\sigma_m\in\cT_m,\;\sigma_n\in\cT_n\nRightarrow\sigma_m\otimes\sigma_n\in\cT_{mn}\,,
\end{align}
which is one of the properties required for the results in~\cite{BP10}.

We show that in contrast to the finite classical case \cite{LM02,BHLP14}, the regularization in Equation \eqref{eq:Main} is needed in general. That is, we provide an explicit example for which the non-regularized relative entropy formula is not an achievable asymptotic error exponent
\begin{align}\label{eq:reg_needed}
\inf_{\substack{\rho\in\cS\\\sigma\in\cT}}D(\rho\|\sigma)>\lim_{\eps\to0}\limsup_{n\to\infty}\zeta_{\cS,\cT}(n,\eps)\,.
\end{align}
In particular, we find that, even for $n\to\infty$, in general
\begin{align}\label{Eq:ConjMeasure}
\frac{1}{n}\inf_{\mu\in\cT}D\Big(\rho^{\otimes n}\Big\|\int\sigma^{\otimes n}\;\mathrm{d}\mu(\sigma)\Big)\neq\inf_{\sigma\in\cT}D(\rho\|\sigma)\,,
\end{align}
thereby providing a counterexample to this conjectured quantum entropy inequality (see \cite[Equation (20)]{BHOS15} for a variant) which holds in the finite classical setting (see, e.g., \cite[Lemma 3.11]{STH16}).\footnote{After completion of the first version of our work, even simpler examples of composite hypothesis testing problems with no single-letter solution were provided in \cite{mosonyi2020error}.} Note that the $\leq$ direction in Equation \eqref{Eq:ConjMeasure} holds trivially.

Nevertheless, there exist non-commutative cases in which the regularization in Equation \eqref{eq:Main} is not needed and we discuss several such examples. In particular, we give an operational interpretation of the relative entropy of coherence in terms of hypothesis testing.

Finally, we apply the techniques developed in this work to strengthen previously known quantum relative entropy lower bounds on the conditional quantum mutual information \cite{FR14,BHOS15,TB15,wilde15,STH16,JRSWW15,SBT16}
\begin{align}
I(A:B|C)_\rho:=H(AC)_\rho+H(BC)_\rho-H(ABC)_\rho-H(C)_\rho
\end{align}
with $H(C)_\rho:=-\tr\left[\rho_C\log\rho_C\right]$ the von Neumann entropy. We find that
\begin{align}\label{eq:cqmi_lowerbound}
I(A:B|C)_\rho\geq\limsup_{n\to\infty}\frac{1}{n}D\Big(\rho_{ABC}^{\otimes n}\Big\|\int\beta_0(t)\;\mathrm{d}t\big(\cI_A\otimes\cR^{[t]}_{C\to BC}(\rho_{AC})\big)^{\otimes n}\Big)
\end{align}
for some universal probability distribution $\beta_0(t)$ and the rotated Petz recovery maps $R^{[t]}_{C\to BC}$ as defined in Section \ref{sec:recovery}. In contrast to the previously known bounds in terms of the quantum relative entropy \cite{STH16,BHOS15}, the recovery map in Equation \eqref{eq:cqmi_lowerbound} takes a specific form only depending on the reduced state on $BC$. Note that the regularization in Equation \eqref{eq:cqmi_lowerbound} cannot go away in relative entropy distance, as recently shown in~\cite{FF17}. We end with an overview how all known recoverability lower bounds on the conditional quantum mutual information compare and argue that Equation \eqref{eq:cqmi_lowerbound} represents the last possible strengthening.

The remainder of the paper is structured as follows. In Section~\ref{sec:mainresult} we prove our main result about composite asymmetric hypothesis testing. This is followed by Section~\ref{sec:examples} where we discuss several concrete examples including an operational interpretation of the relative entropy of coherence, as well as its R\'enyi analogues in terms of the Petz divergences~\cite{petz_statbook} and the sandwiched relative entropies~\cite{muller2013quantum,WWY14}. In Section~\ref{sec:recovery} we prove the refined lower bound on the conditional mutual information from Equation \eqref{eq:cqmi_lowerbound} and use it to show that the regularization in Equation \eqref{eq:Main} is needed in general. Finally, we end in Section~\ref{sec:discussion} with a discussion of some open questions.

\section{Proof of main result}\label{sec:mainresult}

In the following we give a proof of our main result Theorem~\ref{thm:main}. We first prove the converse, meaning the $\leq$ direction of Theorem~\ref{thm:main}, which follows from the following proposition.

\begin{proposition}\label{lem:converse_lemma}
For $\rho\in\cS$, $\mu\in\cT$, and $\eps\in(0,1)$ we have
\begin{align}
\zeta_{\cS,\cT}(n,\eps) \leq \inf_{\substack{\rho\in\cS\\\mu\in\cT}} \frac{1}{n}\frac{D\left(\rho^{\otimes n}\middle\|\sigma_n(\mu)\right)+1}{1-\eps}\,.\label{eq:main-converse2}
\end{align}
\end{proposition}
\begin{proof}
We follow the original converse proof of quantum Stein's lemma~\cite{HP91} for the states $\rho^{\otimes n}$ and $\sigma_n(\mu)$. By the monotonicity of the quantum relative entropy~\cite{lindblad75}, we have for the measurement $\{M_n,(1-M_n)\}$ that
\begin{align}
&D\left(\rho^{\otimes n}\middle\|\sigma_n(\mu)\right)\nonumber\\
&\geq \tr\left[M_n\rho^{\otimes n}\right]\log\frac{\tr\left[M_n\rho^{\otimes n}\right]}{\tr\left[M_n\sigma_n(\mu)\right]}+\left(1-\tr\left[M_n\rho^{\otimes n}\right]\right)\log\frac{1-\tr\left[M_n\rho^{\otimes n}\right]}{1-\tr\left[M_n\sigma_n(\mu)\right]}\\ 
&\geq-\log2-\tr \left[M_n\rho^{\otimes n}\right]\log\tr\left[M_n\sigma_n(\mu)\right]\\
&\geq-1-\inf_{\rho\in\cS}\tr\left[M_n\rho^{\otimes n}\right]\log\sup_{\mu\in\cT}\tr\left[M_n\sigma_n(\mu)\right]\\
&\geq-1-(1-\eps)\log\sup_{\mu\in\cT}\tr\left[M_n\sigma_n(\mu)\right]
\end{align}
leading to 
\begin{align}
-\frac1n\log\sup_{\mu\in\cT}\tr\left[M_n\sigma_n(\mu)\right]\leq \frac1n \frac{D\left(\rho^{\otimes n}\middle\|\sigma_n(\mu)\right) +1}{1-\eps}
\end{align}
for any $\rho\in\cS$, $\mu\in\cT$, and $0\ll M_n \ll1$ such that $\sup_{\rho\in\cS}\tr\left[(1-M_n)\rho^{\otimes n}\right]\leq\eps$. Taking the supremum over all such $M_n$ and then the infimum over $\rho\in\cS$ and $\mu\in\cT$ leads to the desired result. 
\end{proof}

By taking the appropriate limits in Proposition \ref{lem:converse_lemma}, we immediately find the converse statements
\begin{align}
\lim_{\eps\rightarrow 0} \limsup_{n\rightarrow\infty} \zeta_{\cS,\cT}(n,\eps) &\leq \limsup_{n\rightarrow\infty} \frac{1}{n} \inf_{\substack{\rho\in\cS\\\mu\in\cT}} D\left(\rho^{\otimes n}\middle\|\sigma_n(\mu)\right) \label{Eq:ConverseLimSup}\\
\lim_{\eps\rightarrow 0} \liminf_{n\rightarrow\infty} \zeta_{\cS,\cT}(n,\eps) &\leq \liminf_{n\rightarrow\infty} \frac{1}{n} \inf_{\substack{\rho\in\cS\\\mu\in\cT}} D\left(\rho^{\otimes n}\middle\|\sigma_n(\mu)\right)\, .\label{Eq:ConverseLimInf}
\end{align}

\begin{remark}
In information-theoretic language Equation~\eqref{eq:main-converse2} represents a weak converse, i.e. the limit $\eps\rightarrow 0$ in Equations~\eqref{Eq:ConverseLimSup} and~\eqref{Eq:ConverseLimInf} is required,  and one might be tempted to derive a strong converse that holds for all $\eps\in(0,1)$ by employing quantum versions of the R\'enyi relative entropies \cite{petz_statbook,muller2013quantum,WWY14} or the smooth max-relative entropy \cite{jain02,Datta09-2}. However, because of the missing $\sigma_m\in\cT_m,\;\sigma_n\in\cT_n\nRightarrow\sigma_m\otimes\sigma_n\in\cT_{mn}$ property, the convergence of aforementioned measures to the quantum relative entropy remains unclear (see, e.g., \cite{Verstraete12,BP10,Datta09,Datta13} for corresponding techniques in the context of quantum hypothesis testing). As such, we leave open the question about a strong converse.
\end{remark}

For the proof of the achievability, meaning the $\geq$ direction in Theorem \ref{thm:main}, the basic idea is to start from the corresponding composite Stein's lemma for classical probability distributions and lift the result to the non-commutative setting by solely using properties of quantum entropy. For that we need the measured relative entropy defined as~\cite{donald86,HP91}
\begin{align}\label{eq:measured}
D_{\mathcal{M}}(\rho\|\sigma):=\sup_{(\mathcal{X},\mathcal{M})}D\Big(\underbrace{\sum_{x\in\mathcal{X}}\tr\left[M_x\rho\right]|x\rangle\langle x|}_{=\,\mathcal{M}(\rho)}\Big\|\underbrace{\sum_{x\in\mathcal{X}}\tr\left[M_x\sigma\right]|x\rangle\langle x|}_{=\,\mathcal{M}(\sigma)}\Big)\,,
\end{align}
where the optimization is over finite sets $\mathcal{X}$ and measurements $\mathcal{M}$ on $\mathcal{X}$ with $\tr\left[M_x\rho\right]$ a measure on $\mathcal{X}$. Henceforth, we write for the classical relative entropy between probability distributions $D(P\|Q)$\,---\,defined via the diagonal embedding of $P$ and $Q$ as on the right-hand side of Equation \eqref{eq:measured}. It is known that we can restrict the a priori unbounded supremum to rank-one projective measurements \cite[Theorem 2]{BFT15}. We now prove the achievability direction in Theorem \ref{thm:main} in several steps and start with an achievability bound in terms of the measured relative entropy.

%\begin{proposition}\label{lem:achievability}
%There exists an $N_\eps\in\mathbb{N}$ such that we have for $n>N_\eps$
%\begin{align}
%\zeta_{\cS,\cT}(n,\eps)\geq\frac{1}{n}\inf_{\substack{\rho\in\cS\\\mu\in\cT}}D\left(\rho^{\otimes n}\|\sigma_n(\mu)\right)-\frac{o(n)}{n}\,.
%\end{align}
%\end{proposition}

\begin{lemma}\label{MRelEnt}
For definitions as above and $\eps\in(0,1)$, we have
\begin{align}
\liminf_{n\to\infty}\zeta_{\cS,\cT}(n,\eps)\geq\sup_{k\in\bN}\frac{1}{k}\inf_{\substack{\nu\in\cS\\\mu\in\cT}}D_{\mathcal{M}}\left(\rho_k(\nu)\middle\|\sigma_k(\mu)\right)\,.
\end{align}
\end{lemma}

\begin{proof}
%Analogous to Remark \ref{remSup} it it sufficient to consider the case where there exist $\rho\in\cS,\sigma\in\cT$ such that $\supp(\rho)\subseteq\supp(\sigma)$, otherwise both sides of above equation become infinite by definition. 
For sets of classical probability distributions $\cS$ and $\cT$, we get from the corresponding commutative achievability result that for $\delta>0$ and $\eps\in(0,1)$, there exists $M_{\eps,\delta}\in\bN$ such that for $m\geq M_{\eps,\delta}$ we have
\begin{align}\label{eq:classical}
\zeta_{\cS,\cT}(m,\eps)\geq\inf_{\substack{P\in\cS\\Q\in\cT}}D(P\|Q)-\delta\,.
\end{align}
This is a special case of \cite[Theorem 2]{BHLP14} and we refer to \cite{LM02} as well as references therein for a general discussion of composite hypothesis testing. Now, the strategy is to first measure the quantum states and then to invoke the classical achievability result from Equation~\eqref{eq:classical} for the resulting probability distributions.

This argument is made precise as follows. The classical case implies the existence of a sequence of tests $(T_{k, m})_{m\in\bN}$ for the discrimination problem between the measured state $\cM_k(\cS^{\otimes k})^{\otimes m}$ and the measured state $\cM_k(\cT^{\otimes k})^{\otimes m}$ with $m\in\bN$, such that
\begin{align}
\sup_{\rho\in\cS} \tr\left[(1 - T_{k, m}) \cM_k( \rho^{\otimes k})^{\otimes m}\right] \leq \eps
\end{align}
for all $m\in\bN$, and
\begin{align}
\lim_{m\to\infty} - \frac1m \log\sup_{\sigma\in\cT} \tr\left[T_{k, m} \cM_k(\sigma^{\otimes k})^{\otimes m}\right]\geq \inf_{\substack{\rho\in\cS\\ \sigma\in\cT}} D\left(\cM_k(\rho^{\otimes k})\middle\| \cM_k(\sigma^{\otimes k}\right)\,.
\end{align}
Hence, for any $\delta>0$, there exists an $m_\delta$ such that for all $m\geq m_\delta$ we have
\begin{align}
 - \frac1m \log\sup_{\sigma\in\cT} \tr\left[T_{k, m} \cM_k(\sigma^{\otimes k})^{\otimes m}\right] \geq \inf_{\substack{\rho\in\cS\\ \sigma\in\cT}} D\left(\cM_k(\rho^{\otimes k}) \middle\| \cM_k(\sigma^{\otimes k}\right) - \delta\,.
\end{align}
Defining $T_n := \big(\cM_k^\dagger\big)^{\otimes m}(T_{k,m})\otimes 1_r$ for $n=km+r$, $r\in\{0,\dots, k-1\}$, we get that 
\begin{align}
\sup_{\rho\in\cS} \tr\left[(1- T_n)\rho^{\otimes n}\right] = \sup_{\rho\in\cS} \tr\left[(1 - T_{k,m}) \cM_k(\rho^{\otimes k})^{\otimes m}\right] \leq \eps
\end{align}
for all $n\in\bN$, and thus
\begin{align}
\zeta_{\cS,\cT}(n,\eps) &\geq -\frac1n \log\sup_{\sigma\in\cT} \tr\left[T_n\sigma^{\otimes n}\right]\\ 
&= -\frac1{km+r}\log\sup_{\sigma\in\cT} \tr\left[T_{k,m} \cM_k(\sigma^{\otimes k})^{\otimes m}\right]\\
&\geq \frac{m}{km+r} \inf_{\substack{\rho\in\cS\\\sigma\in\cT}} D\left(\cM_k(\rho^{\otimes k})\middle\| \cM_k(\sigma^{\otimes k})\right) - \frac{m}{km+r} \delta
\end{align}
whenever $n\geq km_\delta$. Therefore, we get
\begin{align}
\liminf_{n\to\infty} \zeta_{\cS,\cT}(n,\eps) \geq \frac{1}{k} \inf_{\substack{\rho\in\cS\\\sigma\in\cT}} D\left(\cM_k(\rho^{\otimes k}) \middle\| \cM_k(\sigma^{\otimes k})\right) - \frac{1}{k} \delta
\end{align}
for any binary POVM $\cM_k$ and $\delta>0$. Taking $\delta\rightarrow 0$ and then the supremum over $\cM_k$ gives
\begin{align}
\liminf_{n\to\infty} \zeta_{\cS,\cT}(n,\eps) &\geq \frac{1}{k} \sup_{\cM_k} \inf_{\substack{\rho\in\cS\\\sigma\in\cT}} D(\cM_k(\rho^{\otimes k}) \| \cM_k(\sigma^{\otimes k})) \\
&\geq \frac{1}{k} \sup_{\cM_k} \inf_{\substack{\nu\in\cS\\\mu\in\cT}} D(\cM_k(\rho_k(\nu)) \| \cM_k(\sigma_k(\mu))) \\ 
&= \frac{1}{k}  \inf_{\substack{\nu\in\cS\\\mu\in\cT}} \sup_{\cM_k} D(\cM_k(\rho_k(\nu)) \| \cM_k(\sigma_k(\mu)))\,,
\end{align}
where the equality follows from Lemma~\ref{applySion}. Since this holds for every $k\in\bN$, we find the claimed
\begin{align}
\liminf_{n\to\infty} \zeta_{\cS,\cT}(n,\eps) &\geq \sup_{k\in\bN} \frac{1}{k}  \inf_{\substack{\nu\in\cS\\\mu\in\cT}} D_\cM(\rho_k(\nu) \| \sigma_k(\mu))\,.
\end{align}
\end{proof}

Next, we argue that the measured relative entropy can in fact be replaced by the quantum relative entropy by only paying an asymptotically vanishing penalty term. For this we need the following lemma, which can be seen as a generalization of the technical argument in the original proof of quantum Stein's lemma~\cite{HP91}.

\begin{lemma}\label{pinching}
Let $\rho_n,\sigma_n\in S\left(\cH^{\otimes n}\right)$ with $\sigma_n$ permutation invariant. Then, we have
\begin{align}
D\big(\rho_n\big\|\sigma_n\big)-\log\poly(n)\leq D_{\mathcal{M}}\big(\rho_n\big\|\sigma_n\big)\leq D\big(\rho_n\big\|\sigma_n\big)\,,
\end{align}
where $\poly(n)$ stands for terms of order at most polynomial in $n$.
\end{lemma}

\begin{proof}
We can restrict ourselves to the case where $\supp\big(\rho_n\big)\subseteq\supp\big(\sigma_n\big)$ since otherwise all relative entropy terms evaluate to infinity by definition. The second inequality follows directly from the definition of the measured relative entropy in Equation \eqref{eq:measured} together with the fact that the quantum relative entropy is monotone \cite{lindblad75}. We now prove the first inequality with the help of asymptotic spectral pinching \cite{hayashi2002}. The pinching map with respect to $\omega\in S(\cH)$ is defined as
\begin{align}
\mathcal{P}_\omega(\cdot):=\sum_{\lambda\in\mathrm{spec}(\omega)} P_\lambda (\cdot)P_\lambda\;\text{with the spectral decomposition $\omega=\sum_{\lambda\in\mathrm{spec}(\omega)}\lambda P_\lambda$.}
\end{align}
Crucially, we have the pinching operator inequality \cite{hayashi2002}
\begin{align}
\mathcal{P}_\omega[X]\gg\frac{X}{|\mathrm{spec}(\omega)|}\,,
\end{align}
where $|\mathrm{spec}(\cdot)|$ denotes the size of the spectrum. From this we can deduce that (see, e.g., \cite[Lemma 4.4]{tomamichel2015quantum})
\begin{align}
D\big(\rho_n\big\|\sigma_n\big)-\log\big|\mathrm{spec}\big(\sigma_n\big)\big|\leq D\big(\mathcal{P}_{\sigma_n}\big(\rho_n\big)\big\|\sigma_n\big) \leq D_{\mathcal{M}}\big(\rho_n\big\|\sigma_n\big)\,,
\end{align}
where the second inequality follows since $\mathcal{P}_{\sigma_n}\big(\rho_n\big)$ and $\sigma_n$ are diagonal in the same basis and the measured relative entropy gives an upper-bound. It remains to show that $\big|\mathrm{spec}\big(\sigma_n\big)\big|\leq\poly(n)$. However, since $\sigma_n$ is permutation invariant we have by Schur-Weyl duality (see, e.g., \cite[Section 5]{harrow_phd}) that in the Schur basis
\begin{align}
\sigma_n=\bigoplus_{\lambda\in\Lambda_n}\sigma_{Q_\lambda}\otimes1_{P_\lambda}\quad\text{with $|\Lambda_n|\leq\poly(n)$ and $\text{dim} \left[\sigma_{Q_\lambda}^0\right]\leq\poly(n)$.}
\end{align}
where $\sigma_{Q_\lambda}^0$ is the projector onto the support of $\sigma_{Q_\lambda}$. This implies the claim.
\end{proof}

By combining Lemma \ref{MRelEnt} together with Lemma \ref{pinching} we find for $\eps\in(0,1)$ that
\begin{align}\label{eq:intermediate}
\liminf_{n\to\infty}\zeta_{\cS,\cT}(n,\eps)\geq\limsup_{n\to\infty}\frac{1}{n}\inf_{\substack{\nu\in\cS\\\mu\in\cT}}D\left(\rho_n(\nu)\|\sigma_n(\mu)\right)\,.
\end{align}
The next step is to argue that asymptotically the infimum over states $\rho_n(\nu)$ can without loss of generality be restricted to iid states $\rho^{\otimes n}$ with $\rho\in\cS$.

\begin{lemma}\label{MinOut}
For definitions as above and $\omega_n\in S\left(\cH^{\otimes n}\right)$, we have
\begin{align}
\frac{1}{n}\inf_{\nu\in\cS}D\left(\rho_n(\nu)\middle\|\omega_n\right)\geq\frac{1}{n}\inf_{\rho\in\cS}D\left(\rho^{\otimes n}\middle\|\omega_n\right)-\frac{2d^2\log(n+1)}{n}\,,
\end{align}
where $d:=\mathrm{dim}\left(\cH\right)$.
\end{lemma}

\begin{proof}
For $\nu\in\cS$ and $H(\rho):=-\tr\left[\rho\log\rho\right]$ the von Neumann entropy, we observe the following chain of arguments 
\begin{align}
&\frac{1}{n} D\left(\rho_n(\nu)\middle\|\omega_n\right)\notag\\
&= \frac{1}{n} D\Big(\sum_{i=1}^Np_i\rho_i^{\otimes n}\Big\|\omega_n\Big)\\
&=-\frac{1}{n}H\Big(\sum_{i=1}^Np_i\rho_i^{\otimes n}\Big)-\frac{1}{n}\sum_{i=1}^Np_i\tr\left[\rho_i^{\otimes n}\log{\omega_n}\right]\\
&\geq-\frac{1}{n}\sum_{i=1}^N p_i H\left(\rho_i^{\otimes n}\right)-\frac{\log{(n+1)^{2d^2}}}{n}-\frac{1}{n}\sum_{i=1}^Np_i\tr\left[\rho_i^{\otimes n}\log{\omega_n}\right]\\
&\geq\min_{\rho_i}\frac{1}{n}D\left(\rho^{\otimes n}_i\middle\|\omega_n\right)-\frac{2d^2\log(n+1)}{n}\\
&\geq\inf_{\rho\in\cS}\frac{1}{n}D\left(\rho^{\otimes n}\middle\|\omega_n\right)-\frac{2d^2\log(n+1)}{n}\,,
\end{align}
where the first equality holds by an application of Carath\'edory's theorem with $N\leq(n+1)^{2d^2}$ (Lemma \ref{carat}), and the first inequality by an almost-convexity property of the von Neumann entropy (Lemma \ref{caratentropy}). All other steps are elementary. Since the above argument holds for all $\nu\in\cS$, the claim follows.
\end{proof}

Lemma \ref{MinOut} together with Equation \eqref{eq:intermediate} gives for $\eps\to0$ that
\begin{align}
\limsup_{n\to\infty}\frac{1}{n}\inf_{\substack{\rho\in\cS\\\mu\in\cT}}D\left(\rho^{\otimes n} \|\sigma_n(\mu)\right) &\leq \sup_{k\in\bN} \frac{1}{k}  \inf_{\substack{\nu\in\cS\\\mu\in\cT}} D_\cM(\rho_k(\nu) \| \sigma_k(\mu)) \\ 
&\leq \lim_{\eps\rightarrow 0} \liminf_{n\to\infty}\zeta_{\cS,\cT}(n,\eps) \\
&\leq \liminf_{n\to\infty}\frac{1}{n}\inf_{\substack{\rho\in\cS\\\mu\in\cT}}D\left(\rho^{\otimes n} \|\sigma_n(\mu)\right)\,,
\end{align}
where the last step follows from Equation~\eqref{Eq:ConverseLimInf}. This shows that the limit
\begin{align}
\lim_{n\to\infty}\frac{1}{n}\inf_{\substack{\rho\in\cS\\\mu\in\cT}}D\left(\rho^{\otimes n} \|\sigma_n(\mu)\right)
\end{align}
exists and all the inequalities above hold as equalities. Furthermore, we have 
\begin{align}
\limsup_{n\to\infty}\frac{1}{n}\inf_{\substack{\rho\in\cS\\\mu\in\cT}}D\left(\rho^{\otimes n} \|\sigma_n(\mu)\right) &\leq \lim_{\eps\rightarrow 0}\limsup_{n\to\infty}\zeta_{\cS,\cT}(n,\eps) \\
&\leq \limsup_{n\to\infty}\frac{1}{n}\inf_{\substack{\rho\in\cS\\\mu\in\cT}}D\left(\rho^{\otimes n} \|\sigma_n(\mu)\right)\,,
\end{align}
which concludes the proof of Theorem~\ref{thm:main}.
\qed

%Together with the converse we get
%\begin{align}
%\liminf_{n\to\infty}\frac{1}{n}\inf_{\substack{\rho\in\cS\\\mu\in\cT}}D\left(\rho^{\otimes n}\|\sigma_n(\mu)\right)\geq\zeta_{\cS,\cT}(0)\geq\limsup_{n\to\infty}\frac{1}{n}\inf_{\substack{\rho\in\cS\\\mu\in\cT}}D\left(\rho^{\otimes n}\|\sigma_n(\mu)\right)\; \\
%\Rightarrow\zeta_{\cS,\cT}(0)=\lim_{n\to\infty}\frac{1}{n}\inf_{\substack{\rho\in\cS\\\mu\in\cT}}D\left(\rho^{\otimes n}\|\sigma_n(\mu)\right)\,,
%\end{align}
%which finishes the proof of Theorem \ref{thm:main}. \qed

%%%%%%%%%%%%%%%%%%%%%%%%%%%%%%%%%%%%%%%%%%%%%%%%%%%%%%%%%%%%%%%%%%%%%%%%%%%%%%%%

\section{Examples and extensions}\label{sec:examples}

Here, we discuss several concrete examples of composite discrimination problems\,---\,some of which have a single-letter solution.

%%%%%%%%%%%%%%%%%%%%%%%%%%%%%%%%%%%%%%%%%%%%%%%%%%%%%%%%%%%%%%%%%%%%%%%%%%%%%%%%

\subsection{Relative entropy of coherence}\label{sec:coherence}

Following the literature around~\cite{Baumgratz14}, the set of states diagonal in a fixed basis $\{|c\rangle\}$ is called incoherent and denoted by $\cC\subseteq S(\cH)$. The relative entropy of coherence of $\rho\in S(\cH)$ is defined as
\begin{align}
D_{\cC}(\rho):=\inf_{\sigma\in\cC} D(\rho\|\sigma)\,.
\end{align}
Based on our main result (Theorem \ref{thm:main}), we can characterize the following discrimination problem.

\begin{description}
\item[{\bf Null hypothesis}:] the fixed state $\rho^{\otimes n}$
\item[{\bf Alternative hypothesis}:] the convex hull of iid coherent states
\begin{align}
\bar{\cC}_n:=\Big\{\int\sigma^{\otimes n}\;\mathrm{d}\mu(\sigma)\Big|\mu\in\cC\Big\}
\end{align}
\end{description}

Namely, Theorem \ref{thm:main} gives
\begin{align}\label{eq:coherence}
\zeta_{\bar{\cC}}(\infty,0):=&\lim_{\eps\to0}\lim_{n\to\infty}\zeta_{\bar{\cC}}(n,\eps) \\
=&\lim_{n\to\infty}\frac{1}{n}\inf_{\mu\in\cC}D\Big(\rho^{\otimes n}\Big\|\int\sigma^{\otimes n}\;\mathrm{d}\mu(\sigma)\Big)\\
=&D_{\cC}(\rho)\,,
\end{align}
where the limit in Equation~\eqref{eq:coherence} exists because the relative entropy of coherence is additive on product states \cite{Chitambar16}, and the last step follows from a general property of the relative entropy of coherence (Lemma \ref{relEntropyGasym}) applied to the decohering channel. In fact, there is even a single-letter solution for the following less restricted discrimination problem.

\begin{description}
\item[{\bf Null hypothesis}:] the fixed state $\rho^{\otimes n}$
\item[{\bf Alternative hypothesis}:] the convex set of coherent states $\cC_n$
\end{description}

It is straightforward to check that this hypothesis testing problem fits the general framework of~\cite{BP10}, leading to
\begin{align}\label{eq:coherence_brandao}
\zeta_{\cC}(\infty,\eps):=\lim_{n\to\infty}\zeta_{\cC}(n,\eps)=\lim_{n\to\infty}\frac{1}{n}\inf_{\sigma_n\in\cC_n}D\left(\rho^{\otimes n}\|\sigma_n\right)=D_{\cC}(\rho)\quad\forall\eps\in(0,1)\,,
\end{align}
where the last step again follows from a general property of the relative entropy of coherence (Lemma \ref{relEntropyGasym}). Thus, we have two a priori different hypothesis testing scenarios that both give an operational interpretation to the relative entropy of coherence. In the following we give a simple self-contained proof of Equation \eqref{eq:coherence_brandao} that is different from the rather involved steps in~\cite{BP10} and instead follows ideas from~\cite{Audenaert2008,HT14}. The goal is the quantification of the optimal asymptotic error exponent
%(as we will see the following limit exists)
\begin{align}
\zeta_{\cC}(n,\eps)&:=-\frac{1}{n}\log\inf_{\substack{0\ll M_n\ll1\\ \tr\left[M_n\rho^{\otimes n}\right]\geq1-\eps}}\sup_{\sigma_n\in\cC_n}\tr\left[M_n\sigma_n\right] \\ 
\mathrm{with}\quad\zeta_{\cC}(\infty,\eps)&:=\lim_{n\to\infty}\zeta_{\cC}(n,\eps)\,.
\end{align}

\begin{proposition}\label{thm:convex_achievability}
For the discrimination problem as above with $\eps\in(0,1)$, we have
\begin{align}
\zeta_{\cC}(\infty,\eps)=D_{\cC}(\rho)\,.
\end{align}
\end{proposition}

Note that Proposition \ref{thm:convex_achievability} is independent of $\text{supp}(\rho)$ as the set $\cC$ includes full rank states. A weak converse for $\eps\to0$ follows exactly as in Lemma \ref{lem:converse_lemma}, together with Lemma \ref{relEntropyGasym} to make the expression single-letter. For the strong converse as claimed in Proposition \ref{thm:convex_achievability}, we make use of a general family of quantum R\'enyi entropies: the Petz divergences~\cite{petz_statbook}. For $\rho,\sigma\in S(\cH)$ and $s\in(0,1)\cup(1,\infty)$ they are defined as
\begin{align}
D_s\left(\rho\middle\|\sigma\right):=\frac{1}{s-1}\log\tr\left[\rho^s\sigma^{1-s}\right]\,,
\end{align}
whenever either $s<1$ and $\rho$ is not orthogonal to $\sigma$ in Hilbert-Schmidt inner product or $s>1$ and the support of $\rho$ is contained in the support of $\sigma$. (Otherwise we set $D_s(\rho\|\sigma):=\infty$.) The corresponding R\'enyi relative entropies of coherence are given by \cite{Chitambar16}
\begin{align}\label{eq:coherence_additivity}
\text{$D_{s,\cC}(\rho):=\inf_{\sigma\in\cC}D_s(\rho\|\sigma)$ with the additivity property $D_{s,\cC}\left(\rho^{\otimes n}\right)=n D_{s,\cC}(\rho)$.}
\end{align}
Using similar standard arguments \cite{Nagaoka05} as in Lemma \ref{lem:converse_lemma} but based on the monotonicity of the Petz divergences, we find for $s\in(1,2]$ that
\begin{align}
&-\frac{1}{n}\log\inf_{0\leq M_n\leq1}\Big\{\tr\left[M_n\sigma_n\right]\Big|\tr\left[(1-M_n)\rho^{\otimes n}\right]\leq\eps\Big\}\notag\\
&\leq\frac{1}{n}\cdot D_s\left(\rho^{\otimes n}\middle\|\sigma_n\right)+\frac{1}{n}\frac{s}{s-1}\frac{1}{\log(1-\eps)}\,.
\end{align}
By taking the infimum over $\sigma_n\in\cC_n$, a basic application of Sion's minimax theorem (Lemma \ref{Sion}), using the additivity from Equation~\eqref{eq:coherence_additivity}, taking the limit $n\to\infty$ as well as the limit \cite{Chitambar16}
\begin{align}\label{eq:newC}
\lim_{s\to 1}D_{s,\cC}(\rho)=D_{\cC}(\rho)\,,
\end{align}
we find the claimed strong converse $\zeta_{\cC}(\infty,\eps)\leq D_{\cC}(\rho)$. The achievability direction of Proposition \ref{thm:convex_achievability} is based on the Petz divergences as well.

\begin{lemma}\label{thm:general_convex}
For the discrimination problems as above with $n\in\mathbb{N}$ and $\eps\in(0,1)$, we have for $s\in(0,1)$ that
\begin{align}\label{eqn:general_convex}
\zeta_{\cC}(n,\eps)\geq D_{s,\cC}(\rho)-\frac{1}{n}\frac{s}{1-s}\log\frac{1}{\eps}\,.
\end{align}
\end{lemma}

Taking the limit $n\to\infty$ as well as the limit $s\to1$ using Equation~\eqref{eq:newC}, we then find the claimed achievability $\zeta_{\cC}(\infty,\eps)\geq D_{\cC}(\rho)$.

\begin{proof}[Proof of Lemma \ref{thm:general_convex}]
It is straightforward to check with Sion's minimax theorem (Lemma \ref{Sion}) that
\begin{align}\label{eq:sion_applied}
\inf_{\substack{0\leq M_n\leq1\\ \tr\left[M_n\rho^{\otimes n}\right]\geq1-\eps}}\sup_{\sigma_n\in\cC_n}\tr\left[M_n\sigma_n\right]=\sup_{\sigma_n\in\cC_n}\inf_{\substack{0\leq M_n\leq1\\ \tr\left[M_n\rho^{\otimes n}\right]\geq1-\eps}}\tr\left[M_n\sigma_n\right]\,.
\end{align}
Now, for $\lambda_n\in\mathbb{R}$ with $n\in\mathbb{N}$ we choose $M_n(\lambda_n):=\left\{\rho^{\otimes n}-2^{\lambda_n}\sigma_n\right\}_+$ where $\{\cdot\}_+$ denotes the projector on the eigenspace of the positive spectrum. We have $0\ll M_n(\lambda_n)\ll1$ and by Audenaert's inequality (Lemma \ref{lem:audenaert}) with $s\in(0,1)$ we get
\begin{align}\label{eq:audenaert_one}
\tr\left[(1-M_n(\lambda_n))\rho^{\otimes n}\right]\leq2^{(1-s)\lambda_n}\tr\left[\left(\rho^{\otimes n}\right)^s\sigma_n^{1-s}\right]=2^{(1-s)\left(\lambda_n-D_s\left(\rho^{\otimes n}\middle\|\sigma_n\right)\right)}\,.
\end{align}
Moreover, again Audenaert's inequality (Lemma \ref{lem:audenaert}) for $s\in(0,1)$ implies
\begin{align}\label{eq:audenaert_two}
\tr\left[M_n(\lambda_n)\sigma_n\right]\leq2^{-s\lambda_n}\tr\left[\left(\rho^{\otimes n}\right)^s\sigma_n^{1-s}\right]=2^{-s\lambda_n-(1-s)D_s\left(\rho^{\otimes n}\middle\|\sigma_n\right)}\,.
\end{align}
Hence, choosing
\begin{align}
\text{$\lambda_n:=D_s\left(\rho^{\otimes n}\middle\|\sigma_n\right)+\log\eps^{\frac{1}{1-s}}$ with $M_n:=M_n(\lambda_n)$}
\end{align}
leads with Equation \eqref{eq:audenaert_one} to $\tr\left[M_n\rho^{\otimes n}\right]\geq1-\eps$. Finally, Equation \eqref{eq:sion_applied} together with Equation \eqref{eq:audenaert_two} and the additivity property from Equation \eqref{eq:coherence_additivity} leads to the claim.
\end{proof}

We note that a more refined analysis of the above calculation allows to determine the Hoeffding bound as well as the strong converse exponent (cf.~\cite{Audenaert2008,HT14}). The former gives an operational interpretation to the R\'enyi relative entropy of coherence $D_{s,\cC}(\rho)$, whereas the latter gives an operational interpretation to the sandwiched R\'enyi relative entropies of coherence~\cite{Chitambar16}
\begin{align}
\tilde{D}_{s,\cC}(\rho):=\inf_{\sigma\in\cC}\tilde{D}_s(\rho\|\sigma)
\end{align} 
with the sandwiched R\'enyi entropies
\begin{align}
\tilde{D}_s(\rho\|\sigma):=\frac{1}{s-1}\log\tr\left[\left(\sigma^{\frac{1-s}{2s}}\rho\sigma^{\frac{1-s}{2s}}\right)^s\right]
\end{align}
whenever either $s<1$ and $\rho$ is not orthogonal to $\sigma$ in Hilbert-Schmidt inner product or $s>1$ and the support of $\rho$ is contained in the support of $\sigma$~\cite{muller2013quantum,WWY14}. (Otherwise we set $D_s(\rho\|\sigma):=\infty$.) The crucial insight for the proof is again the additivity property $\tilde{D}_{s,\cC}\left(\rho^{\otimes n}\right)=n\tilde{D}_{s,\cC}(\rho)$, that was already shown in~\cite{Chitambar16}.

%%%%%%%%%%%%%%%%%%%%%%%%%%%%%%%%%%%%%%%%%%%%%%%%%%%%%%%%%%%%%%%%%%%%%%%%%%%%%%%%

\subsection{Relative entropy of recovery}\label{sec:hypothesis_recovery}

The relative entropy of recovery of $\rho_{ABC}\in S(\cH_{ABC})$ and its regularized version are defined as~\cite{SW14,BHOS15,BFT15}\footnote{This limit exists and is finite as for $a_n:=D(A;B|C)_{\rho^{\otimes n}}\geq0$ we have the monotonicity property $a_{n+m}\leq a_n+a_m$.}
\begin{align}
&D(A;B|C)_\rho:=\inf_{\cR} D\big(\rho_{ABC}\big\|(\mathcal{I}_{A}\otimes\cR_{C\to BC})\left(\rho_{AC}\right)\big)\\ 
\mathrm{and}\quad &D^{\infty}(A;B|C)_\rho:=\lim_{n\to\infty}\frac{1}{n}D(A;B|C)_{\rho^{\otimes n}}\,,
\end{align}
where the infimum goes over all completely positive and trace preserving maps $\cR_{C\to BC}$. It was recently shown that in general \cite{FF17}
\begin{align}
D^{\infty}(A;B|C)_\rho\neq D(A;B|C)_\rho\,.
\end{align}
Using the framework from~\cite{BP10}, the following discrimination problem was linked to the regularized relative entropy of recovery~\cite{CHMOSWW16}.

\begin{description}
\item[{\bf Null hypothesis}:] the fixed state $\rho_{ABC}^{\otimes n}$
\item[{\bf Alternative hypothesis}:] for any $\cR_{C^n\to B^nC^n}$ completely positive and trace preserving, the convex set of states
\begin{align}
\cR^n:=\left\{(\cI_{A^n}\otimes\cR_{C^n\to B^nC^n})\left(\rho_{AC}^{\otimes n}\right)\right\}
\end{align}
\end{description}

Namely, for $\eps\in(0,1)$ we have for the corresponding asymptotic error exponent
\begin{align}
\zeta_{\cR}(\infty,\eps):=\lim_{n\to\infty}\zeta_{\cR}(n,\eps)=D^{\infty}(A;B|C)_\rho\,.
\end{align}
In contrast, our main result (Theorem \ref{thm:main}) covers the following discrimination problem.

\begin{description}
\item[{\bf Null hypothesis}:] the fixed state $\rho_{ABC}^{\otimes n}$
\item[{\bf Alternative hypothesis}:] for any $\cR_{C\to BC}$ completely positive and trace preserving, the convex hull of iid states
\begin{align}
\bar{\cR}^n:=\Big\{\int\left((\cI_A\otimes\cR_{C\to BC})(\rho_{AC})\right)^{\otimes n}\;\mathrm{d}\mu(\cR)\Big\}\,.
\end{align}
\end{description}

%Namely, Theorem~\ref{thm:main} gives the asymptotic error exponent
%\begin{align}
%\limsup_{n\to\infty} \frac{1}{n} \inf_{\mu\in\cR} D\Big(\rho_{ABC}^{\otimes n}\Big\|\int\big((\cI_A\otimes\cR_{C\to BC})(\rho_{AC})\big)^{\otimes n}\;\mathrm{d}\mu(\cR)\Big)\,. 
%\end{align}

Interestingly, we can show that the asymptotic error exponents of the two discrimination problems are actually identical.

\begin{proposition}
With the definitions as above, we have
\begin{align}
\lim_{n\to\infty}&\frac{1}{n}\inf_{\cR} D\big(\rho_{ABC}^{\otimes n}\big\|(\cI_A\otimes\cR_{C^n\to B^nC^n})\left(\rho_{AC}^{\otimes n}\right)\big) \nonumber\\
&= \lim_{n\to\infty} \frac{1}{n} \inf_{\mu\in\cR} D\Big(\rho_{ABC}^{\otimes n}\Big\|\int\big((\cI_A\otimes\cR_{C\to BC})(\rho_{AC})\big)^{\otimes n}\;\mathrm{d}\mu(\cR)\Big)\,.
\end{align}
\end{proposition}

\begin{proof}
One direction of the inequality is by definition and for the other direction we use a de Finetti reduction for quantum channels \cite[Lemma 8]{BHOS15} that was first derived in \cite{FR14}. Namely, we have for $\omega_{C^n}\in S\left(\cH_C^{\otimes n}\right)$ and permutation invariant $\cR_{C^n\to B^nC^n}$ that
\begin{align}
\cR_{C^n\to B^nC^n}\left(\omega_{C^n}\right)\ll\poly(n)\int\left(\cR_{C\to BC}\right)^{\otimes n}\left(\omega_{C^n}\right)\mathrm{d}\nu(\cR)
\end{align}
for some measure $\mathrm{d}\nu(\cR)$ over the completely positive and trace preserving maps on $C\to BC$. As explained in the proof of~\cite[Proposition 9]{BHOS15}, the joint convexity of the quantum relative entropy together with the operator monotonicity of the logarithm then imply that
\begin{align}
&D\left(\rho_{ABC}^{\otimes n}\middle\|\cR_{C^n\to B^nC^n}\left(\rho_{AC}^{\otimes n}\right)\right)\qquad \nonumber\\ 
&\geq D\Big(\rho_{ABC}^{\otimes n}\Big\|\int\big((\cI_A\otimes\cR_{C\to BC})(\rho_{AC})\big)^{\otimes n}\;\mathrm{d}\nu(\cR)\Big)-\log\poly(n)\,.
\end{align}
\end{proof}

As such, we can conclude that
\begin{align}
\zeta_{\bar{\cR}}(\infty,0) &:=\lim_{\eps\to0}\liminf_{n\to\infty}\zeta_{\bar{\cR}}(n,\eps) \nonumber\\
&= \lim_{\eps\to0}\limsup_{n\to\infty}\zeta_{\bar{\cR}}(n,\eps) = D^{\infty}(A;B|C)_\rho\,. \label{Eq:RelRecovExpo}
\end{align}

%The equality of the state discrimination rates then follows immediately.
%\begin{corollary}
%For the discrimination problems as above we have $\zeta_{A:B|C}(0)=\bar{\zeta}_{A:B|C}(0)$.
%\end{corollary}

%%%%%%%%%%%%%%%%%%%%%%%%%%%%%%%%%%%%%%%%%%%%%%%%%%%%%%%%%%%%%%%%%%%%%%%%%%%%%%%%

\subsection{Quantum mutual information}\label{ssec:QMI}

The quantum mutual information of $\rho_{AB}\in S(\cH_{AB})$ is defined as
\begin{align}
I(A:B)_\rho:=H(A)_\rho+H(B)_\rho-H(AB)_\rho\,.
\end{align}
Our main result from Section \ref{sec:mainresult} provides a solution to the following discrimination problem.

\begin{description}
\item[{\bf Null hypothesis}:] the fixed state $\rho_{AB}^{\otimes n}$
\item[{\bf Alternative hypothesis}:] the convex hull of iid states
\begin{align}
\bar{\cT}_{A^n:B^n}:=\Big\{\rho_A^{\otimes n}\otimes\int\sigma_B^{\otimes n}\;\mathrm{d}\mu(\sigma)\Big|\mu\in S(\cH_B)\Big\}\,.
\end{align}
\end{description}

Namely, we have
\begin{align}
\bar{\zeta}_{A:B}(\infty,0):=&\lim_{\eps\to0}\lim_{n\to\infty}\bar{\zeta}_{A:B}(n,\eps)\\
=&\lim_{n\to\infty}\frac{1}{n}\inf_{\mu\in\bar{\cT}}D\Big(\rho^{\otimes n}_{AB}\Big\|\rho_A^{\otimes n}\otimes\int\sigma_B^{\otimes n}\;\mathrm{d}\mu(\sigma)\Big)\\
=&I(A:B)_\rho\,.
\end{align}
Here, the last equality follows from the easily checked identity
\begin{align}\label{eq:mutualinfo_identity}
I(A:B)_\rho=\inf_{\sigma_B\in S(\cH)}D(\rho_{AB}\|\rho_A\otimes\sigma_B)\,.
\end{align}
More general composite discrimination problems leading to the quantum mutual information were solved in~\cite{HT14} and in the following we further extend these results (cf.~the classical work \cite{TH15}).

\begin{description}
\item[{\bf Null hypothesis}:] the fixed state $\rho_{AB}^{\otimes n}$
\item[{\bf Alternative hypothesis}:] the set of states
\begin{align}
\cT_{A^n:B^n}:=\left\{\sigma_{A^n}\otimes\sigma_{B^n}\in S\left(\cH_{AB}^{\otimes n}\right)\middle|\sigma_{A^n}\;\text{or}\;\sigma_{B^n}\;\text{permutation invariant}\right\}\,.
\end{align}
\end{description}

The goal is again the quantification of the optimal asymptotic error exponent
%(as we will see the following limit exists)
\begin{align}
&\zeta_{A:B}(n,\eps):=-\frac{1}{n}\log\inf_{\substack{0\ll M_n\ll1\\ \tr\left[M_n\rho^{\otimes n}\right]\geq1-\eps}}\sup_{\sigma_{A^n}\otimes\sigma_{B^n}\in\cT_n}\tr\left[M_{A^nB^n}\sigma_{A^n}\otimes\sigma_{B^n}\right]\label{eq:error_mutual}\\
&\mathrm{with}\quad\zeta_{A:B}(\infty,\eps):=\lim_{n\to\infty}\zeta_{A:B}(n,\eps)\,.
%\quad\mathrm{and}\quad\zeta_{A:B}(\infty,0):=\lim_{\eps\to0}\zeta_{A:B}(\eps)
\end{align}
Note that the sets $\cT_{A^nB^n}$ are not convex and hence the minimax technique used in Section \ref{sec:coherence} does not work here. However, following the ideas in~\cite{HT14,TH15} we can exploit the permutation invariance and use de Finetti reductions of the form~\cite{Hayashi2009,christandl09} to find the following.

\begin{proposition}\label{prop:quantum_mutual}
For the discrimination problem as above with $\eps\in(0,1)$, we have
\begin{align}
\zeta_{A:B}(\infty,\eps)=I(A:B)_\rho\,.
\end{align}
\end{proposition}

The achievability direction is based on the following lemma.

\begin{lemma}\label{eq:mutual-new}
For the discrimination problem as above with $n\in\mathbb{N}$ and $\eps\in(0,1)$, we have for $s\in(0,1)$ that
\begin{align}
\zeta_{A:B}(n,\eps)\geq\inf_{\sigma\in S(\cH)}D_s\left(\rho_{AB}\middle\|\sigma_A\otimes\sigma_B\right)-\frac{1}{n}\frac{s}{1-s}\log\frac{1}{\eps}-\frac{\log\poly(n)}{n}\,.
\end{align}
\end{lemma}

\begin{proof}
Without loss of generality assume that $\sigma_{A^n}$ is permutation invariant. We choose
\begin{align}
M_{A^nB^n}(\lambda_n):=\left\{\rho_{AB}^{\otimes n}-2^{\lambda_n}\omega_{A^n}\otimes\omega_{B^n}\right\}_+\quad \nonumber\\ 
\mathrm{with}\quad\omega_{A^n}:={n+|A|^2-1\choose n}^{-1}\tr_{\tilde{A}^n}\left[P^{\mathrm{Sym}}_{A^n\tilde{A}^n}\right]\,,
\end{align}
where $P^{\mathrm{Sym}}_{A^n\tilde{A}^n}$ denotes the projector onto the symmetric subspace of $\cH_A^{\otimes n}\otimes\cH_{\tilde{A}}^{\otimes n}$ with $|A|=|\tilde{A}|$ (denoting the dimension of $\cH_A$ by $|A|$), and similarly for $B^n$. Audenaert's inequality (Lemma \ref{lem:audenaert}) gives that
\begin{align}
\tr\left[(1-M_{A^nB^n}(\lambda_n))\rho^{\otimes n}_{AB}\right]&\leq2^{(1-s)\lambda_n}\tr\left[\left(\rho^{\otimes n}_{AB}\right)^s\left(\omega_{A^n}\otimes\omega_{B^n}\right)^{1-s}\right] \nonumber \\ &\leq2^{(1-s)\left(\lambda_n-\inf_{\sigma_{A^n}\otimes\sigma_{B^n}\in\cT_n}D_s\left(\rho^{\otimes n}_{AB}\middle\|\sigma_{A^n}\otimes\sigma_{B^n}\right)\right)}\,.
\end{align}
Furthermore, we have by Schur-Weyl duality that $\sigma_{A^n}\leq{n+|A|^2-1\choose n}\,\omega_{A^n}$ for all permutation invariant $\sigma_{A^n}$ (see, e.g., \cite[Lemma 1]{HT14}) and thus again by Audenaert's inequality (Lemma \ref{lem:audenaert})
\begin{align}
&\tr\left[M_{A^nB^n}(\lambda_n)\left(\sigma_{A^n}\otimes\sigma_{B^n}\right)\right] \\ 
&=\tr\left[M_{A^nB^n}(\lambda_n)\left(\sigma_{A^n}\otimes\left(\sum_{\pi\in S_n}U_{B^n}(\pi)\sigma_{B^n}U_{B^n}^\dagger(\pi)\right)\right)\right]\;\text{($S_n$: symm. group)}\notag\\
&\leq\underbrace{{n+|A|^2-1\choose n}{n+|B|^2-1\choose n}}_{=:\;p(n)\;\leq\;\poly(n)}\tr\left[M_{A^nB^n}(\lambda_n)\left(\omega_{A^n}\otimes\omega_{B^n}\right)\right]\notag\\
&\leq p(n)\cdot2^{-s\lambda_n}\tr\left[\left(\rho^{\otimes n}_{AB}\right)^s\left(\omega_{A^n}\otimes\omega_{B^n}\right)^{1-s}\right]\notag\\
&\leq p(n)\cdot2^{-s\lambda_n-(1-s)\inf_{\sigma_{A^n}\otimes\sigma_{B^n}\in\cT_n}D_s\left(\rho^{\otimes n}_{AB}\middle\|\sigma_{A^n}\otimes\sigma_{B^n}\right)}\,.\label{eq:mutual_error-last}
\end{align}
We now choose
\begin{align}
\text{$\lambda_n:=\inf_{\sigma_{A^n}\otimes\sigma_{B^n}\in\cT_n}D_s\left(\rho^{\otimes n}_{AB}\middle\|\sigma_{A^n}\otimes\sigma_{B^n}\right)+\log\eps^{\frac{1}{1-s}}$ with $M_{A^nB^n}:=M_{A^nB^n}(\lambda_n)$,}
\end{align}
from which we get $\tr\left[M_{A^nB^n}\rho_{AB}^{\otimes n}\right]\geq1-\eps$ and together with Equation \eqref{eq:error_mutual} and Equation \eqref{eq:mutual_error-last} that
\begin{align}
\zeta_{A:B}^n(\eps)\geq\inf_{\sigma_{A^n}\otimes\sigma_{B^n}\in\cT_n}D_s\left(\rho_{AB}^{\otimes n}\middle\|\sigma_{A^n}\otimes\sigma_{B^n}\right)-\frac{1}{n}\frac{s}{1-s}\log\frac{1}{\eps}-\frac{\log p(n)}{n}\,.
\end{align}
To deduce the claim it is now sufficient to argue that the R\'enyi quantum mutual information\footnote{This definition is slightly different from the R\'enyi quantum mutual information discussed in~\cite{HT14}.}
\begin{align}
I_s(A:B)_\rho:=\inf_{\sigma_A\otimes\sigma_B\in S(\cH)}D_s\left(\rho_{AB}\middle\|\sigma_A\otimes\sigma_B\right)
\end{align}
is additive on tensor product states. This, however, follows exactly as in the classical case~\cite[App.~A-C]{TH15} from the (quantum) Sibson identity~\cite[Lemma 3]{SW13}
\begin{align}\label{sibson}
D_s\left(\rho_{AB}\middle\|\sigma_A\otimes\sigma_B\right)=D_s\left(\rho_{AB}\middle\|\sigma_A\otimes\bar{\sigma}_B\right)+D_s\left(\bar{\sigma}_B\middle\|\sigma_B\right)\quad\\ 
\mathrm{with}\quad\bar{\sigma}_B:=\frac{\left(\tr_A\left[\rho_{AB}^s\sigma_A^{1-s}\right]\right)^\frac{1}{s}}{\tr\left[\left(\tr_A\left[\rho_{AB}^s\sigma_A^{1-s}\right]\right)^\frac{1}{s}\right]}\,.\nonumber
\end{align}
\end{proof}

Taking the limit $n\to\infty$ in Lemma \ref{eq:mutual-new} and then taking the limit $s\to1$ via the quantum Sibson identity from Equation~\eqref{sibson} and Equation~\eqref{eq:mutualinfo_identity} yields
\begin{align}
\lim_{s\to1}I_s(A:B)_\rho=I(A:B)_\rho\,,
\end{align}
gives the claimed achievability $\zeta_{A:B}(\infty,0)\geq I(A:B)_\rho$. A weak converse for $\eps\to0$ follows as in Lemma \ref{lem:converse_lemma} and the strong converse as claimed in Proposition \ref{prop:quantum_mutual} is derived similarly as in Proposition \ref{thm:convex_achievability}\,---\,by noting that it is sufficient to prove a converse for testing
\begin{align}
\text{$\rho_{AB}^{\otimes n}$ against $\rho_A^{\otimes n}\otimes\sigma_{B^n}$.}
\end{align}

%together with Lemma \ref{eq:mutualinfo_identity} to make the expression single-letter. For the strong converse as claimed in Proposition \ref{prop:quantum_mutual}, we find for $s\in(1,2]$ that \cite{Nagaoka05}
%\begin{align}
%&-\frac{1}{n}\log\inf_{0\leq M_n\leq1}\Big\{\tr\left[M_n(\sigma_{A_n}\otimes\sigma_{B_n})\right]\Big|\tr\left[(1-M_n)\rho^{\otimes n}_{AB}\right]\leq\eps\Big\}\notag\\
%&\leq\frac{1}{n} D_s\left(\rho^{\otimes n}_{AB}\middle\|\sigma_{A_n}\otimes\sigma_{B_n}\right)+\frac{1}{n}\frac{s}{s-1}\frac{1}{\log(1-\eps)}\,.
%\end{align}
%By taking the infima over $\sigma_n\in\cC_n$ and a basic application of Sion's minimax theorem (Lemma \ref{Sion}), using the additivity from Equation~\eqref{}, taking the limit $n\to\infty$, and taking the limit $s\to1$ we find the claimed strong converse $\zeta_{\cC}(\infty,\eps)\leq D_{\cC}(\rho)$.

A more refined analysis of the above calculation along the work~\cite{HT14} allows to determine the Hoeffding bound for the product testing discrimination problem as above. However, for the strong converse exponent we are missing the additivity of the sandwiched R\'enyi quantum mutual information
\begin{align}
\tilde{I}_s(A:B)_\rho:=\inf_{\sigma_ A\otimes\sigma_B\in S(\cH)}\tilde{D}_s\left(\rho_{AB}\middle\|\sigma_A\otimes\sigma_B\right)
\end{align}
on product states.

%%%%%%%%%%%%%%%%%%%%%%%%%%%%%%%%%%%%%%%%%%%%%%%%%%%%%%%%%%%%%%%%%%%%%%%%%%%%%%%%

\section{Conditional quantum mutual information}\label{sec:recovery}

Here, we discuss how our results are related to the conditional quantum mutual information. This allows us to show that the regularization in our formula for composite asymmetric hypothesis testing as stated in Theorem \ref{thm:main} is needed in general.

%%%%%%%%%%%%%%%%%%%%%%%%%%%%%%%%%%%%%%%%%%%%%%%%%%%%%%%%%%%%%%%%%%%%%%%%%%%%%%%%

\subsection{Recoverability bounds}

The following is a proof of the lower bound on the conditional quantum mutual information from Equation \eqref{eq:cqmi_lowerbound}.

\begin{theorem}\label{thm:cqmi}
For $\rho_{ABC}\in S(\cH_{ABC})$ we have
\begin{align}\label{eq:cqmi}
I(A:B|C)_\rho\geq\limsup_{n\to\infty}\frac{1}{n}D\Big(\rho_{ABC}^{\otimes n}\Big\|\int\beta_0(t)\left(\cI_A\otimes\cR^{[t]}_{C\to BC}(\rho_{AC})\Big)^{\otimes n}\mathrm{d}t\right)\,,
\end{align}
where $\cR^{[t]}_{C\to BC}(\cdot):=\rho_{BC}^{\frac{1+it}{2}}\big(\rho_C^{\frac{-1-it}{2}}(\cdot)\rho_C^{\frac{-1+it}{2}}\big)\rho_{BC}^{\frac{1-it}{2}}$ with the inverses understood as generalized inverses and $\beta_0(t):=\frac{\pi}{2}\left(\cosh(\pi t)+1\right)^{-1}$.
\end{theorem}

\begin{proof}
We start from the lower bound~\cite[Theorem 4.1]{SBT16} applied to $\rho_{ABC}^{\otimes n}$ (with the support conditions taken care of as in the corresponding proof)
\begin{align}
I(A:B|C)_\rho=\frac{1}{n}I\left(A^n:B^n\middle|C^n\right)_{\rho^{\otimes n}}\geq\frac{1}{n}D_{\mathcal{M}}\left(\rho_{ABC}^{\otimes n}\middle\|\sigma_{A^nB^nC^n}\right)
\end{align}
with
\begin{align}\label{eq:sigma}
\text{$\sigma_{A^nB^nC^n}:=\int\beta_0(t)\left(\sigma_{ABC}^{[t]}\right)^{\otimes n}\mathrm{d}t$ and $\sigma_{ABC}^{[t]}:=\left(\cI_A\otimes \cR^{[t]}_{C\to BC}\right)(\rho_{AC})$,}
\end{align}
where we have used that the conditional quantum mutual information is additive on product states. Now, we simply observe that $\sigma_{A^nB^nC^n}$ is permutation invariant and hence the claim can be deduced from Lemma \ref{pinching} together with taking the limit superior $n\to\infty$.
\end{proof}

Together with previous work we find the following corollary that encompasses all known recoverability lower bounds on the conditional quantum mutual information.

\begin{corollary}\label{cor:cqmi}
For $\rho_{ABC}\in S(\cH_{ABC})$ the conditional quantum mutual information $I(A:B|C)_\rho$ is lower bounded by
\begin{align}
&-\int\beta_0(t)\log\Big\|\sqrt{\rho_{ABC}}\sqrt{\sigma_{ABC}^{[t]}}\Big\|_1^2\;\mathrm{d}t\\
&D_{\mathcal{M}}\Big(\rho_{ABC}\Big\|\int\beta_0(t)\sigma_{ABC}^{[t]}\;\mathrm{d}t\Big)\\
&\limsup_{n\to\infty}\frac{1}{n}D\Big(\rho_{ABC}^{\otimes n}\Big\|\int\beta_0(t)\big(\sigma_{ABC}^{[t]}\big)^{\otimes n}\mathrm{d}t\Big)
\end{align}
with $\sigma_{ABC}^{[t]}$ from Equation~\eqref{eq:sigma}.
\end{corollary}

The first bound was shown in~\cite[Section 3]{JRSWW15}, the second one in~\cite[Theorem 4.1]{SBT16}, and the third one is Theorem \ref{thm:cqmi}. We note that the lower bounds are typically strict in the non-commutative case, as can be seen from numerical work (see, e.g., \cite{BHOS15}). In contrast to the second and third bound, the first lower bound is not tight in the commutative case but has the advantage that the average over $\beta_0(t)$ stands outside of the distance measure used. Moreover, the distribution $\beta_0(t)$ cannot be taken outside the relative entropy measure in the second and the third bound, since quantum Stein's lemma would then lead to a contradiction to a recent counterexample from~\cite[Section 5]{FF17}. Namely, there exists $\theta\in\left[0,\pi/2\right]$ such that
\begin{align}
&I(A:B|C)_\rho\ngeq\inf_{\cR}D\left(\rho_{ABC}\middle\|(\cI_A\otimes\cR_{C\to BC})(\rho_{AC})\right)\;\label{eq:example_fawzi} \end{align}
for the pure state $\rho_{ABC}=|\rho\rangle\langle\rho|_{ABC}$ with
\begin{align}
|\rho\rangle_{ABC}=&\frac{1}{\sqrt{2}}\big(\cos(\theta)|0\rangle_A\otimes|1\rangle_C+\sin(\theta)|1\rangle_A\otimes|0\rangle_C\big)\otimes|1\rangle_B\notag\\
&+\frac{1}{\sqrt{2}}|0\rangle_A\otimes|0\rangle_B\otimes|0\rangle_C\,.
\end{align}
It seems that the only remaining conjectured strengthening is the lower bound in terms of the non-rotated Petz map~\cite[Section 8]{bertawilde14}
\begin{align}
I(A:B|C)_\rho\geq-\log\Big\|\sqrt{\rho_{ABC}}\sqrt{\sigma_{ABC}^{[0]}}\Big\|_1^2\,.
\end{align}
We refer to~\cite{Lemm17} for the latest progress in that direction.

The arguments in this section can also be applied to lift the strengthened monotonicity from~\cite[Corollary 4.2]{SBT16}. For $\rho\in S(\cH)$, $\sigma$ a positive semi-definite operator on $\cH$, and $\mathcal{N}$ a completely positive trace preserving map on the same space this leads to
\begin{align}
D(\rho\|\sigma)-D(\mathcal{N}(\rho)\|\mathcal{N}(\sigma))\geq\limsup_{n\to\infty}\frac{1}{n}D\Big(\rho^{\otimes n}\Big\|\int\beta_0(t)\left(\cR^{[t]}_{\sigma,\mathcal{N}}(\rho)\right)^{\otimes n}\mathrm{d}t\Big),
\end{align}
where $\cR^{[t]}_{\sigma,\mathcal{N}}(\cdot):=\sigma^{\frac{1+it}{2}}\mathcal{N}^{\dagger}\left(\mathcal{N}(\sigma)^{\frac{-1-it}{2}}(\cdot)\mathcal{N}(\sigma)^{\frac{-1+it}{2}}\right)\sigma^{\frac{1-it}{2}}$. Together with~\cite[Section 3]{JRSWW15} and~\cite[Corollary 4.2]{SBT16} we then again have the three lower bounds as in Corollary \ref{cor:cqmi}.

%%%%%%%%%%%%%%%%%%%%%%%%%%%%%%%%%%%%%%%%%%%%%%%%%%%%%%%%%%%%%%%%%%%%%%%%%%%%%%%%

\subsection{Regularization necessary}

Here, we use our bound on the conditional quantum mutual information (Theorem \ref{thm:cqmi}) to show that the regularization in Theorem \ref{thm:main} is in general needed (see also~\cite{BDKSSS05}). That is, we give a proof for Equation \eqref{eq:reg_needed}. Namely, by Theorem \ref{thm:cqmi} we have\footnote{Alternatively, we could employ the implicitly stated bound \cite[Equation 38]{BHOS15}.}
\begin{align}
I(A:B|C)_\rho&\geq\limsup_{n\to\infty}\frac{1}{n}D\Big(\rho_{ABC}^{\otimes n}\Big\|\int\beta_0(t)\left(\cI_A\otimes\cR^{[t]}_{C\to BC}(\rho_{AC})\right)^{\otimes n}\mathrm{d}t\Big)\\
&\geq\lim_{n\to\infty}\frac{1}{n}\inf_{\mu\in\cR}D\Big(\rho_{ABC}^{\otimes n}\Big\|\int\left(\cI_A\otimes\cR_{C\to BC}\left(\rho_{AC}\right)\right)^{\otimes n}\mathrm{d}\mu(\cR)\Big)\,.
\end{align}
From the second composite discrimination problem described in Section \ref{sec:hypothesis_recovery} we see that the latter quantity is equal to the asymptotic error exponent $\zeta_{\bar{\cR}}(\infty,0)$ as given in Equation~\eqref{Eq:RelRecovExpo} for testing
\begin{align}
\text{$\rho_{ABC}^{\otimes n}$ against $\int\left((\cI_A\otimes\cR_{C\to BC})(\rho_{AC})\right)^{\otimes n}\;\mathrm{d}\mu(\cR)$.}
\end{align}
Now, if the regularization in the asymptotic formula for $\zeta_{\bar{\cR}}(\infty,0)$ would actually not be needed this would imply that
\begin{align}
I(A:B|C)_\rho\geq\inf_{\cR}D\left(\rho_{ABC}\|(\cI_A\otimes\cR_{C\to BC})(\rho_{AC})\right)\,.
\end{align}
However, this is in contradiction with the counterexample from~\cite[Section 5]{FF17} as discussed in Equation \eqref{eq:example_fawzi}. Hence, we conclude that the regularization for composite asymmetric quantum hypothesis testing is needed in general. \qed

%%%%%%%%%%%%%%%%%%%%%%%%%%%%%%%%%%%%%%%%%%%%%%%%%%%%%%%%%%%%%%%%%%%%%%%%%%%%%%%%

\section{Conclusion}\label{sec:discussion}

We extended quantum Stein's lemma in asymmetric quantum hypothesis testing by showing that the optimal asymptotic error exponent for testing convex combinations of quantum states $\rho^{\otimes n}$ against convex combinations of quantum states $\sigma^{\otimes n}$ is given by a regularized quantum relative entropy formula which does not become single-letter in general. Moreover, we gave various examples when our formula as well as extensions thereof become single-letter. It remains interesting to find more non-commutative settings that allow for single-letter solutions.

Another related problem is that of symmetric hypothesis testing, where it is well-known that in the case of fixed iid states $\rho^{\otimes n}$ against $\sigma^{\otimes n}$ the optimal asymptotic error exponent is given by the quantum Chernoff bound~\cite{ACMBMAV07,NS09}
\begin{align}
C(\rho,\sigma):=\sup_{0\leq s\leq1}-\log\tr\left[\rho^s\sigma^{1-s}\right]\,.
\end{align}
For this symmetric setting, it was conjectured in~\cite{AM14} that for finite sets $\cS$ and $\cT$ the corresponding composite asymptotic error exponent is given by
\begin{align}\label{eq:conjCB}
C(\cS,\cT):=\inf_{\substack{\rho\in\cS\\\sigma\in\cT}}C(\rho,\sigma)\,,
\end{align}
with definitions analogue to those given earlier for the asymmetric setting. However, it was recently shown that already in the setting of a fixed null hypothesis $\cS=\{\rho\}$ above conjecture is in general false~\cite{mosonyi2020error}.\footnote{See, however, \cite{Li16} for a related problem that does allow for an exact single-letter characterisation.}

Moreover, one can again consider testing convex combinations of iid states $\rho^{\otimes n}$ with $\rho\in\cS$ against convex combinations of iid states $\sigma^{\otimes n}$ with $\sigma\in\cT$. Similarly to our work for the asymmetric setting, we then have that the following rate for the asymptotic error exponent is achievable (assuming that the limit exists)
\begin{align}\label{eq:symmetric_regularized}
\sup_{0\leq s\leq 1} \lim_{n\to\infty}\frac{1}{n}\inf_{\substack{\nu\in\cS\\\mu\in\cT}}-\log\tr\Big[\Big(\int\rho^{\otimes n}\;\mathrm{d}\nu(\rho)\Big)^s\Big(\int\sigma^{\otimes n}\;\mathrm{d}\mu(\sigma)\Big)^{1-s}\Big]\,.
\end{align}
However, it was already shown in~\cite{hiai2009quantum} that this does in general not simplify to the single-letter form in Equation \eqref{eq:conjCB}. We refer to~\cite[Section I]{mosonyi2020error} for an excellent overview of the recent progress on composite hypothesis testing.

Finally, we note that finding single-letter achievability results for composite hypothesis testing problems has important applications in network quantum Shannon theory~\cite[Section 5.2]{QWW17}.

%%%%%%%%%%%%%%%%%%%%%%%%%%%%%%%%%%%%%%%%%%%%%%%%%%%%%%%%%%%%%%%%%%%%%%%%%%%%%%%%

\section*{Acknowledgements} 
We thank an anonymous referee for extensive feedback and pointing out detailed solutions to multiple errors in previous versions of this manuscript. This work was completed prior to MB and FB joining the AWS Center for Quantum Computing. CH acknowledges support from the VILLUM FONDEN via the QMATH Centre of Excellence (Grant no. 10059), the Spanish MINECO, project FIS2013-40627-P, FIS2016-80681-P (AEI/FEDER, UE) and FPI Grant No.~BES-2014-068888, as well as by the Generalitat de Catalunya, CIRIT project no.~2014-SGR-966.

%%%%%%%%%%%%%%%%%%%%%%%%%%%%%%%%%%%%%%%%%%%%%%%%%%%%%%%%%%%%%%%%%%%%%%%%%%%%%%%%

\appendix

\section{Some Lemmas}

Here, we give several lemmas that are used in the main part. We start with Sion's minimax theorem \cite{sion58}.

\begin{lemma}\label{Sion}
Let $X$ be a compact convex subset of a linear topological space and $Y$ a convex subset of a linear topological space. If a real-valued function on $X\times Y$ is such that
\begin{itemize}
\item[$\diamond$] $f(x,\cdot)$ is upper semi-continuous and quasi-concave on $Y$ for every $x\in X$
\item[$\diamond$] $f(\cdot,y)$ is lower semi-continuous and quasi-convex on $X$ for every $y\in Y$\,,
\end{itemize}
then we have
\begin{align}
\min_{x\in X} \sup_{y\in Y} f(x,y) = \sup_{y\in Y} \min_{x\in X} f(x,y)\,.
\end{align}
\end{lemma}

%If the measured relative entropy is optimized over closed, convex sets then Sion's minimax theorem can be applied.

The following is a special case of \cite[Lemma 13]{BHLP14}, which is based on a more involved minimax theorem taking into account the possibility that the relative entropy can be infinite.

\begin{lemma}\label{applySion}
Let $\cS,\cT\subseteq S(\cH)$ be closed, convex sets. Then, we have
\begin{align}
\min_{\substack{\rho\in\cS\\ \sigma\in\cT}}D_{\mathcal{M}}(\rho\|\sigma)=\sup_{(\mathcal{X},\mathcal{M})}\min_{\substack{\rho\in\cS\\ \sigma\in\cT}}D\Big(\sum_{x\in\mathcal{X}}\tr\left[M_x\rho\right]|x\rangle\langle x|\Big\|\sum_{x\in\mathcal{X}}\tr\left[M_x\rho\right]|x\rangle\langle x|\Big)\,.
\end{align}
\end{lemma}

%\begin{proof}
%This is basically~\cite[Lemma 20]{BHLP14} but the result stated there is wrong in full generality even though our case should be covered.
%\end{proof}

We have the following discretization result.

\begin{lemma}\label{carat}
For every probability measure $\mu$ on the Borel $\sigma$-algebra of $\cS\subseteq S(\cH)$ with the dimension of $\cH$ given by $d$, there exists a probability distribution $\{ p_i\}_i^N$ with $N\leq (n+1)^{2d^2}$ and $\rho_i\in\cS$ such that
\begin{align}
\int \rho^{\otimes n}\;\mathrm{d}\mu(\rho)=\sum_{i=1}^N p_i\rho_i^{\otimes n}\,.
\end{align}
\end{lemma}

\begin{proof}
The idea is to use Carath\'eodory theorem together with the smallness of the symmetric subspace. For pure states the proof from~\cite[Corollary D.6]{BCR11} applies and the general case follows immediately by considering purifications and taking the partial trace over the purifying system.
\end{proof}

The von Neumann entropy has the following almost-convexity property (besides its well-known concavity).

\begin{lemma}\label{caratentropy}
Let $\rho_i\in S(\cH)$ for $i=1,\ldots,N$ and $\{p_i\}$ be a probability distribution. Then, we have
\begin{align}
H\Big(\sum_{i=1}^N p_i\rho_i\Big)\leq \sum_{i=1}^N p_i H(\rho_i) + \log N\,.
\end{align}
\end{lemma}

\begin{proof}
This follows from elementary quantum entropy inequalities (see, e.g., \cite[Chapter 11]{NC00})
\begin{align}
H\Big(\sum_{i=1}^N p_i\rho_i\Big)\leq\sum_{i=1}^N p_i H\left(\rho_i\right)+H(p_i)\leq \sum_{i=1}^N p_i H(\rho_i)+\log N\,.
\end{align}
\end{proof}

%The relative entropy between classical-quantum states can be written as follows.
%\begin{lemma}\cite[Lemma 1]{piani09}\label{relentcq}
%Let $\rho_i,\sigma_i\in S(\cH)$ for $i=1,\ldots,N$ and $\{p_i\},\{q_i\}$ be probability distributions. Then, we have
%\begin{align}
%H\left(\sum_{i=1}^Np_i\ketbra{}{i}{i}\otimes\rho_i\middle\|\sum_{i=1}^Nq_i\ketbra{}{i}{i}\otimes\sigma_i\right)=D\left(\{p_i\}\middle\|\{q_i\}\right)+\sum_{i=1}^Np_iH\left(\rho_i\middle\|\sigma_i\right)\,.
%\end{align}
%\end{lemma}

The following is a property of the quantum relative entropy \cite[Theorem 3]{GMS09}.

\begin{lemma}\label{relEntropyGasym}
Let $\mathcal{N}$ be a trace-preserving, completely positive map with $\mathcal{N}(1)=1$ (unital) and $\mathcal{N}^2=\mathcal{N}$ (idempotent). Then, the minimum relative entropy distance between $\rho\in S(\cH)$ and $\sigma\in S(\cH)$ in the image of $\mathcal{N}$ satisfies 
\begin{align}
\inf_{\sigma\in\mathrm{Im}(\mathcal{N})}D(\rho\|\sigma)=H(\mathcal{N}(\rho))-H(\rho)=D(\rho\|\mathcal{N}(\rho))\,.
\end{align}
In particular, we have for the relative entropy of coherence $D_{\cC}(\rho)=D(\rho\|\rho_{\mathrm{diag}})$, where $\rho_{\mathrm{diag}}$ denotes the state obtained from $\rho$ by deleting all off-diagonal elements.
\end{lemma}

Audenaert's inequality originally used to derive the quantum Chernoff bound can be stated as follows \cite[Theorem 1]{ACMBMAV07}.

\begin{lemma}\label{lem:audenaert}
Let $X,Y\gg0$ and $s\in(0,1)$. Then, we have
\begin{align}
\tr\left[X^sY^{1-s}\right]\geq\tr\left[X\left(1-\left\{X-Y\right\}_+\right)\right]+\tr\left[Y\left\{X-Y\right\}_+\right]\,.
\end{align}
\end{lemma}

%%%%%%%%%%%%%%%%%%%%%%%%%%%%%%%%%%%%%%%%%%%%%%%%%%%%%%%%%%%%%%%%%%%%%%%%%%%%%%%%

%\bibliographystyle{abbrv}
\bibliographystyle{arxiv_no_month}
\bibliography{Bib}

\end{document}